\documentclass[12pt]{article}

\usepackage{amssymb}
\usepackage{amsfonts}
\usepackage{amsmath}
\usepackage{latexsym}
\usepackage{amsthm}
\usepackage{epsfig}
\usepackage{bm}
\usepackage{color}

\newcommand{\ignore}[1]{}

\makeatletter
\renewcommand{\subsubsection}{\@startsection{subsubsection}{3}{0pt}{-12pt}{-5pt}{\normalsize\bf}}
\makeatother

\newtheorem{claim}{Claim}[section]
\newtheorem{proposition}[claim]{Proposition}
\newtheorem{observation}[claim]{Observation}

\newtheorem{lemma}[claim]{Lemma}
\newtheorem{theorem}{Theorem}[section]
\newtheorem{definition}{Definition}
\newtheorem{corollary}[claim]{Corollary}

\newtheorem{fact}[claim]{Fact}
\newtheorem*{theorem*}{Theorem}
\newtheorem*{claim*}{Claim}
\newtheorem*{lemma*}{Lemma}
\newtheorem*{proposition*}{Proposition}

\renewcommand{\Pr}{{\mathrm{Pr}}}

\newcommand{\bits}{\{0,1\}}

\newcommand{\bitsn}{\bits^n}

\newcommand{\poly}{\mathrm{poly}}

\newcommand{\supp}{{\mathrm{supp}}}

   \newcommand{\rnote}[1]{{#1}}

\title{Noisy population recovery in polynomial time }

\date{\today}
\author{Anindya De\thanks{Some part of this work was done while the author was a postdoc at DIMACS, Rutgers.} \\
Northwestern University \\
Evanston, IL, USA \\
{\tt anindya@eecs.northwestern.edu}
\and
Michael Saks \thanks{Supported by NSF grant CCF-1218711 and by  Simons Foundation award 332622.} \\
Department of Mathematics \\
Rutgers University \\
Piscataway, NJ, USA \\
{\tt saks@math.rutgers.edu}
\and
Sijian Tang \thanks{Supported by NSF grant  CCF-1218711} \\
Department of Mathematics \\
Rutgers University \\
Piscataway, NJ, USA \\
{\tt st509@math.rutgers.edu}
}
\begin{document}

\bibliographystyle{alpha}
\maketitle

\begin{abstract}
In the noisy population recovery problem of Dvir et al.~\cite{dvir2012restriction}, the goal is to learn
an unknown  distribution $f$ on binary strings of length $n$ from noisy samples. For some parameter $\mu \in [0,1]$,
a noisy sample is generated by flipping each coordinate of a sample from $f$  independently  with
probability $(1-\mu)/2$.
We assume an upper bound $k$ on the size of the support of the distribution, and the
goal is to estimate the probability of any string to within some given error $\varepsilon$.  It is known
that the algorithmic complexity and sample complexity of this problem are polynomially related to each other.

We show that for $\mu >   0$, the sample complexity (and hence the algorithmic complexity)
is bounded by a polynomial in $k$, $n$ and  $1/\varepsilon$    improving upon the  previous best result of $\mathsf{poly}(k^{\log\log k},n,1/\varepsilon)$ due to Lovett and Zhang~\cite{lovett2015improved}.

Our proof combines ideas from \cite{lovett2015improved} with a \emph{noise attenuated} version of M\"{o}bius inversion. In turn, the latter crucially uses the construction of \emph{robust local inverse} due to Moitra and Saks~\cite{moitra2013polynomial}.

\end{abstract}

\section{Introduction}

\subsection{Background and Our Result}

The population recovery problem is a basic problem in noisy unsupervised learning which has received significant attention in the recent past~\cite{dvir2012restriction, wigderson2012population, moitra2013polynomial, lovett2015improved}. In this problem,
there is an unknown distribution $f$ over binary strings of length $n$, and an error parameter $0<\mu<1$.
Noisy samples from it are generated as:
\begin{itemize}
\item  Choose a string $x$ according to $f$.
\item  Flip each coordinate of $x$ independently with probability $\frac{1-\mu}{2}$.
\end{itemize}
Given access to these noisy samples, the task of the learner is to output a set of strings $S$ and for each string $x$ in $S$, an estimate $\tilde{f}(x)$ of $f(x)$, such that $|\tilde{f}(x)-f(x)|\leq \epsilon$. And for all $x\notin S$, $f(x)\leq \epsilon$. For $\mu=1$, the problem is trivial to solve, whereas for $\mu=0$, the distribution $f$ cannot be recovered with any number of samples. As $\mu$ becomes smaller, the learning problem becomes progressively harder.
There is an alternate (and easier) model called the lossy model where instead of flipping bits, each bit is replaced by a '?' independently with probability $1-\mu$ and presented to the learner. \\
This problem was introduced  by Dvir et al. \cite{dvir2012restriction} who related it to the problem of learning DNF from restrictions. For the lossy model,  Dvir et al. \cite{dvir2012restriction} gave a polynomial time algorithm for population recovery for any $\mu\gtrsim 0.365$. Their analysis was improved by Batman, et al. \cite{batman2013finding} who showed that the same algorithm works for any $\mu>1-1/\sqrt{2}\approx 0.293$.  Subsequently, Moitra and Saks \cite{moitra2013polynomial}  gave a polynomial time algorithm for population recovery in the lossy model for any $\mu>0$.


For the noisy sample problem, algorithms are known only when the support size of $f$ is bounded by a parameter $k$.
Wigderson and Yehudayoff \cite{wigderson2012population} developed a framework called ``partial identification" and used this to give an algorithm that runs in time  $\mathsf{poly}(k^{\log k},n,1/\epsilon)$ for any $\mu>0$. They also showed that their framework cannot obtain algorithms running in time better than $\mathsf{poly}(k^{\log \log k})$.

 Lovett and Zheng \cite{lovett2015improved} improved on this to show that
the time complexity of this problem is at most $\mathsf{poly}(k^{\log \log k},n,1/\epsilon)$ for any $\mu>0$. Interestingly, while their algorithm matches the lower bound in \cite{wigderson2012population}, their algorithm departs from the framework of \cite{wigderson2012population}.  This offers the possibility that one might be able to achieve better algorithms by extending the techniques of \cite{lovett2015improved}.  Another interesting feature of this problem  is that the algorithmic complexity of the problem is polynomial in the sample complexity of the problem. This seems to have been first explicitly mentioned in \cite{lovett2015improved}  though they refer to \cite{batman2013finding, moitra2013polynomial}.  Thus, it suffices to focus on bounding the sample complexity of the noisy population recovery problem (which is a purely information theoretic quantity).

In this paper, we improve on the results of  \cite{lovett2015improved} and show that for any $\mu>0$, the time complexity of noisy population recovery problem is at most $\mathsf{poly}(k,n,1/\epsilon)$.
This is the first polynomial time algorithm
for any $\mu<1$. The following is our main theorem.

\begin{theorem}\label{thm:1}
For any $\mu>0$, there exists an algorithm for the noisy population recovery problem, running in time $\mathsf{poly}((k/\epsilon)^{{O}_{\mu}(1)},n)$. Here $O_{\mu}(1) = \tilde{O}(1/\mu^4)$. \end{theorem}
For the ensuing discussion, we first fix some preliminaries.

\subsection{Preliminaries}
In this section, we include some basic preliminaries concerning Fourier expansion and noise operators.
Let $f: \{0,1\}^n \rightarrow \mathbb{R}$. Recall that any such $f$
can be expressed uniquely as a linear combination of {\em characters}, where for $S \subseteq [n]$, the character $\chi_S(x)$
is equal to $\prod_{i \in S}(-1)^{x_i}$.    For $S \subseteq [n]$, \rnote{the {\em Fourier coefficient}}
$\widehat{f}(S)$ is defined to be $\widehat{f}(S)=\sum_{x \in \bitsn} f(x)\chi_S(x)$. With this definition, it follows that $f(x) =  \sum_{S \subseteq [n]} f_S \cdot \chi_S(x)$ where
$f_S = 2^{-n} \cdot |\widehat{f}(S)|$. We define $\Vert f\Vert_1 = \sum_x |f(x)|$ and $\Vert \widehat{f} \Vert_{L_1} = 2^{-n} \sum_S |\widehat{f}(S)| = \sum_S |f_S|$. Also we define the support of the Fourier spectrum be $\mathsf{supp}(\widehat{f})=\{S:\widehat{f}(S)\not= 0\}$.
~\\
Let $\mathcal{F}$ be the space of real-valued functions on $\{0,1\}^n$. For $S \subseteq [n]$, we define the operator $X_S : \mathcal{F} \rightarrow \mathcal{F}$ as,
$$
(X_S  f)(x) =f(x) \cdot \chi_S(x),
$$
where $f \in \mathcal{F}$. Next, we define the Bonami-Beckner noise operator.  For $\mu>0$ and $i\in [n] $,  define $T_{\mu ,i} :\mathcal{F} \to \mathcal{F} $ to be the operator that only adds noise  in coordinate $i$. In other words,
$$
(T_{\mu, i}f)(x)= \frac{1+\mu}{2}\cdot f(x)+\frac{1-\mu}{2}\cdot f(x^i),
$$
where $x^i$ is the element obtain by flipping the $i$-th bit of x.
For $S\subset [n] $ define the operator $T_{\mu,S}:\mathcal{F}\to\mathcal{F} $ to be the tensor product of $T_{\mu,i}$ for $i \in S$.  In other words,
$$
(T_{\mu, S}f)(x)= \sum_{T \subseteq S} \frac{(1+\mu)^{|T|} \cdot (1-\mu)^{|S \setminus T|}}{2^{|S|}} \cdot f(x^T),
$$
where $x^T$ is obtained by flipping $x$ in the coordinates in $T$. We define the Bonami-Beckner operator $T_{\mu} : \mathcal{F} \rightarrow \mathcal{F}$ as $T_{\mu} = T_{\mu, [n]}$. Another way to define the action of $T_{\mu}$ is the following. Let $D_{\mu}$ be the product distribution on $\{0,1\}^n$ such that for all $i \in [n]$, $\Pr[e_i=0]=(1+\mu)/2$, $\Pr[e_i=1]=(1-\mu)/2$. Then,
$$
(T_{\mu} f)(x)=\mathbb{E}_{e\sim D_{\mu}}[f(x+e)].
$$
This definition implies $|(T_{\mu} f )(x)|\leq 1,\forall x$.
\subsubsection{Robust local inverse for the noise matrix}
Let us define the matrix $A_{\mu, n} \in \mathbb{R}^{(n+1) \times (n+1)}$  as
$$
A_{\mu, n}(i,j) = \binom{i}{j} \cdot \mu^j \cdot (1-\mu)^{i-j}.
$$
We index $[n+1]$ by $0 \le i \le n$ and $\binom{i}{j}$ is defined to be $0$ if $j>i$. In a key part of the paper, we will use the following key theorem from \cite{moitra2013polynomial}.
\begin{theorem}\label{thm:moitra-saks}\emph{(Moitra-Saks~\cite{moitra2013polynomial})}
For any $\epsilon>0$, there exists $v \in \mathbb{R}^{n+1}$ such that $\Vert A_{\mu,n} \cdot v - e_0 \Vert_{\infty} \le \epsilon$, $\Vert v \Vert_{\infty} \le (2/\epsilon)^{(1/\mu) \cdot \log(2/\mu)}$ and the zeroth coordinate of $A_{\mu,n} \cdot v$ is $1$. Here $e_0 \in \mathbb{R}^{n+1}$ denotes the unit vector with $1$ at the zeroth coordinate. Further, $v$ can be computed in time $\mathsf{poly}(n)$.
\end{theorem}
The non-trivial aspect of the above theorem is that while $A_{\mu,n}$ has very small singular values and as a result, $\Vert A_{\mu,n}^{-1} \cdot e_0 \Vert_{\infty}$ can be exponentially large in $n$, by settling for an $\epsilon$-approximate inverse, it is possible to achieve a significantly better bound. Unfortunately, Theorem~\ref{thm:moitra-saks} is not exactly stated in these words in \cite{moitra2013polynomial} though it follows very easily from the results there.  In Appendix~\ref{app:MS}, we sketch the details on how to obtain Theorem~\ref{thm:moitra-saks} from the results in  \cite{moitra2013polynomial}.

\subsubsection{M\"{o}bius inversion}  Let $(P, \preceq)$ be a poset.  Let $\mathcal{F}_P$ be the space of real-valued functions on $P$. Define $\mu: P  \times P \rightarrow \mathbb{R}$ recursively as follows:
$$
\textrm{For } x\in P, \ \mu(x,x)=1.
$$ $$
\textrm{For } x, y \in P,  \ \mu(x,y) = \mathbf{1}_{x \preceq y}   \cdot \bigg( \sum_{x \preceq z \prec y} -\mu(x,z)\bigg).
$$
We define $\boldsymbol{\zeta}: \mathcal{F}_P \rightarrow \mathcal{F}_P$ and $\mathbf{\mu}: \mathcal{F}_P \rightarrow \mathcal{F}_P$ as
$$
(\boldsymbol{\zeta} f) (x) = \sum_{x \preceq y} f(y) \ \textrm{ and }  \ (\boldsymbol{\mu} f) (x) = \sum_{x \preceq y} \mu (x,y) \cdot f(y).
$$
It is well known (see \cite{Stanley:86}) that the transforms  $\boldsymbol{\zeta}$ and $\boldsymbol{\mu}$ are inverses of each other.
$\boldsymbol{\mu}$ is usually referred to as the M\"{o}bius transform of the poset $P$. While $\boldsymbol{\zeta}$ is always well-conditioned (i.e. $\Vert \boldsymbol{\zeta} \Vert_{1 \rightarrow \infty} \le 1$), the same  is not always true for $\boldsymbol{\mu}$. In particular, the entries of matrix defined by $\mu$ can be exponentially large in $|P|$.
However, in this paper, we consider a special kind of poset for which $\Vert \boldsymbol{\mu} \Vert_{1 \rightarrow \infty}$ is bounded.  To state the next proposition, we will require the following definition. \begin{definition} For $x \in \mathcal{P}([n])$,  define $x^{\downarrow} = \{y : y \preceq x\}$. For $C \subseteq \mathcal{P}([n])$, define $C^{\downarrow}  = \cup_{x \in C} x^{\downarrow}$. We read $C^{\downarrow}$ as the ``downset" generated by $C$.  Also, if $C$ is a set such that for any $x \in C$, the set $x^{\downarrow} \subseteq C$, we say $C$ is ``downward closed". Note that since the underlying poset is $\mathcal{P}([n])$, for $x, y \in \mathcal{P}([n])$, $x \preceq y$ is equivalent to $x \subseteq y$.
\end{definition}

\begin{proposition}\label{prop:mobius}
Let $C \subseteq \mathcal{P}([n])$ and $C^{\downarrow} = \{y : \exists x \in C, \ y \preceq x\}$. Consider the poset defined by $C^{\downarrow}$ ordered by set inclusion. Then,
the M\"{o}bius transform $\boldsymbol{\mu}$ for this poset is defined by
$$
(\boldsymbol{\mu} f) (x)  = \sum_{x \preceq y} (-1)^{| y \setminus x| } \cdot f(y).
$$
\end{proposition}
\begin{proof}
Since it is obvious that $\boldsymbol{\zeta}$ and $\boldsymbol{\mu}$ are invertible transforms, we will just verify that $\boldsymbol{\mu} \circ \boldsymbol{\zeta} : f  \mapsto f$.  To see this,
\begin{eqnarray*}
(\boldsymbol{\mu} \circ \boldsymbol{\zeta} f)(x)  &=& \sum_{x \preceq y} (-1)^{| y \setminus x| } \cdot (\boldsymbol{\zeta}f)(y), \\
&=&  \sum_{x \preceq y} (-1)^{| y \setminus x| } \sum_{y \preceq z} f(z) = \sum_{x \preceq y \preceq z} (-1)^{| y \setminus x| } \cdot f(z).
\end{eqnarray*}
Since the set $C^{\downarrow}$ is downward closed,  it is easy to see that for any $z \succ x$, the sum $\sum_{x \preceq y \preceq z} (-1)^{| y \setminus x| }=0$.  This implies that
$(\boldsymbol{\mu} \circ \boldsymbol{\zeta} f)(x)  = f(x)$ which proves the claim.

\end{proof}

\section{Proof Overview}
We recall that the samples available to learner are obtained in the following manner:
First, an element of $\{0,1\}^n$ is sampled according to $f$ and then each coordinate is flipped independently with probability $(1-\mu)/2$. In other words,
 we have observations from the distribution $T_{\mu} f$ and we want to obtain an estimate of $f$.
Dvir, et al.~\cite{dvir2012restriction} gave a reduction to the case that there is a known subset $X$ of size $2k$
that contains the support of the distribution; for convenience we rescale parameters so that $|X|=k$. Thus, from now onwards, we can assume that we know the support
of $f$ and our task is to estimate the weight assigned by $f$ to these points.

Let us assume that the   support of
$f$ is $\{x_1, \ldots, x_k\}$. To prove Theorem~\ref{thm:1}, it suffices to give an algorithm to compute $f(x_1)$ up to error $\epsilon$. We first show that without loss of generality, we can assume that $x_1=0$ (i.e. the origin). To see this, note that the distribution  $x_1 \oplus T_{\mu} f$ is the same as $T_{\mu} g$ where $g(x) = f(x \oplus x_1)$.

A very basic observation concerning $T_{\mu} f$  is that $\widehat{f}(S)$ can be computed efficiently from $T_{\mu} f$ as long as $|S|$ is small. The following claim formalizes this. For the rest of the discussion, let $\gamma(x,m,z)$ be defined as  $\gamma(x, m, z)  = \mathsf{poly}((1/x)^m, z)$.
\begin{claim}\label{clm:Fourier-compute}
For $S \subset [n]$, $\widehat{f}(S)$ can be computed to additive accuracy $\epsilon$ with probability $1-\kappa$ using $\gamma(\mu, |S|, \epsilon) \cdot \log (1/\kappa)$ samples from $T_{\mu} f$ in time
$n \cdot \gamma(\mu, |S|, \epsilon) \cdot \log (1/\kappa)$.
\end{claim}
\begin{proof}
Observe that $\widehat{f}(S) =\mu^{-|S|} \cdot \mathbf{E}_{x \sim T_{\mu} f} \chi(x)$. Using the fact that $|\chi(x)| \le 1$ and  applying Chernoff bound, we get the claim.
\end{proof}
Thus, for any fixed $\mu>0$, as long as $|S| = O_\mu(\log k)$, the time and sample complexity of computing $\widehat{f}(S)$ using samples from $T_{\mu} f$ is bounded by $\mathsf{poly}(n,k,1/\epsilon)$.
 For $f : \{0,1\}^n \rightarrow \mathbb{R}$, define $\mathsf{Att}(f) \subseteq \mathcal{F}$ as
 $$
 \mathsf{Att}(f) = \{g \in \mathcal{F} \ | E \subseteq \{0,1\}^n \ \textrm{and} \ \ g : x \mapsto f(x) \cdot (T_{\mu} \cdot \mathbf{1}_E)(x) \}.
 $$
 A key step in our algorithm is to generalize Claim~\ref{clm:Fourier-compute} to show that we can compute $\widehat{g}(S)$ for $g \in \mathsf{Att}(f)$ from (sample access to) $T_{\mu} f$ with the same sample complexity as Claim~\ref{clm:Fourier-compute}.
 \begin{claim}\label{clm:att-compute}
 Let $g: \{0,1\}^n \rightarrow \mathbb{R}$ be defined as $g (x) = f(x) \cdot (T_{\mu} \mathbf{1}_E)(x)$. For $S \subset [n]$, $\widehat{g}(S)$ can be computed to additive accuracy $\epsilon$ with probability $1-\delta$ using $\gamma(\mu, |S|, \epsilon) \cdot \log (1/\kappa)$ samples from $T_{\mu} f$ in time
$n \cdot \gamma(\mu, |S|, \epsilon) \cdot \log (1/\kappa)$. Here, we assume that $\mathbf{1}_{E}(\cdot)$ can be efficiently computed.
 \end{claim}
 The proof of Claim~\ref{clm:att-compute} relies heavily on ideas from \cite{lovett2015improved}. We prove this in Section~\ref{sec:att-compute}.  As a consequence, we have the following corollary.

  \begin{corollary}\label{corr:att-compute}
 Let $\ell : \{0,1\}^n \rightarrow \mathbb{R}$ where $T = \Vert \widehat{\ell} \Vert_{L_1}$, $S_0 = \max_{S : \widehat{\ell}(S) \not =0} |S|$.
%
%
%
%
%
 and
$g: \{0,1\}^n \rightarrow \mathbb{R}$ be  defined as Claim~\ref{clm:att-compute}. We assume that $\ell  = \sum_S \ell_S \cdot \chi_S(x)$ where $\ell$ is specified by the list $\{\ell_S\}$.  Then,
  $\langle g , \ell \rangle$ can be computed to accuracy $\epsilon$ with probability $1-\kappa$   using $\gamma(\mu, S_0, \epsilon/T) \cdot \log (|\mathsf{supp}(\widehat{\ell})| /\kappa)$ samples from $T_{\mu} f$ in time $\gamma(\mu, |S_0|, \epsilon/T) \cdot n \cdot |\mathsf{supp}(\widehat{\ell})| \cdot \log (|\mathsf{supp}(\widehat{\ell})| /\kappa)$.
  \end{corollary}
\begin{proof}
 Note that $\langle \ell, g \rangle = \langle \sum_{S} \ell_S \cdot \chi_S(x) , g \rangle = \sum_S \ell_S \cdot \langle \chi_S, g \rangle = \sum_S \ell_S \cdot \widehat{g}(S)$. Claim~\ref{clm:att-compute} implies that using $\gamma(\mu, S_0, \epsilon/T) \cdot \log (|\mathsf{supp}(\widehat{\ell})| /\delta)$ samples from $T_{\mu} f$,   $\widehat{g}(S)$ can be computed  for any $S \in \mathsf{supp}(\widehat{\ell})$ to accuracy $\epsilon/T$ with confidence $1-\kappa/T$. As a corollary,  using $\gamma(\mu, S_0, \epsilon/T) \cdot \log (|\mathsf{supp}(\widehat{\ell})| /\kappa)$ samples from $T_{\mu} f$,  we can compute $\widehat{g}(S)$ for all $S \in \mathsf{supp}(\widehat{\ell})$ to accuracy $\epsilon/T$ with confidence $1-\kappa$. Further, the time complexity of this algorithm is  $\gamma(\mu, S_0, \epsilon/T) \cdot n \cdot |\mathsf{supp}(\widehat{\ell})| \cdot \log (|\mathsf{supp}(\widehat{\ell})| /\kappa)$. As a result,  $\langle \ell, g \rangle $ can be computed to accuracy $\epsilon$ in the claimed time and sample complexity.

\end{proof}

Recall that our task is to compute $f(0)$ to accuracy $\epsilon$. The way we use Corollary~\ref{corr:att-compute}  is as follows: By choosing $E \subseteq \{0,1\}^n$ and $\ell : \{0,1\}^n \rightarrow \mathbb{R}$ carefully, one can ensure that  $\langle \ell, g \rangle \approx (T_{\mu} \mathbf{1}_E)(0) \cdot f(0)$.  Thus, if we can can (approximately) compute $\langle \ell, g \rangle$, we will obtain an approximation for $f(0)$.  The function  $\ell$ is chosen so that $S_0 = O_{\mu}(\log(k/\epsilon))$, $T = (k/\epsilon)^{O_{\mu}(1)}$ and $|\mathsf{supp}(\widehat{\ell})| = (k/\epsilon)^{O_{\mu}(1)}$. Observe that if we plug in the values of $S_0$, $T$ and $|\mathsf{supp}(\widehat{\ell})|$ in Corollary~\ref{corr:att-compute}, the sample and time complexity of computing $\langle \ell, g \rangle$ (to error $\epsilon$) is $(k/\epsilon)^{O_{\mu}(1)}$ and time complexity is $\mathsf{poly}(n, (k/\epsilon)^{O_{\mu}(1)})$. The precise details of this calculation is given in Section~\ref{sec:maintheorem}.

For the moment, we elaborate on how the set $E$ and the function $\ell(\cdot)$ are chosen.
$E \subseteq \{0,1\}^n$ is chosen so that $T_{\mu} \mathbf{1}_E (\cdot)$ has the following properties:
$T_{\mu} \mathbf{1}_E (0) \ge 1/2$ and $T_{\mu} \mathbf{1}_E(x)$ decays exponentially as  $x$ moves away from the origin.
The following lemma makes this precise.
\begin{lemma}\label{lem:lovettlem}
Let $ \{x_1, \ldots, x_k\}$ where $x_1 =0$. Define the set $\mathsf{Far} = \{x_i : d_H(x_1, x_i) \ge (1/\mu^2) \cdot \log k \}$ and define the set $E = \{y \in \{0,1\}^n: d_H(x_1, y) \le d_H(x_i, y) \ \textrm{for all } \ x_i \in \mathsf{Far} \}$.  
\begin{itemize}
\item $(T_{\mu} \mathbf{1}_E)(0) \ge 1/2$.
\item For $x_i \in \mathsf{Far}$, $(T_{\mu} \mathbf{1}_E)(x_i) \le e^{-\frac{1}{2} \cdot \mu^2 \cdot d_H(x_1, x_i)}$.
\end{itemize}
Clearly, the function $\mathbf{1}_E(\cdot)$ can be computed in time $\mathsf{poly}(n,k)$. Further, $(T_{\mu} \mathbf{1}_E)(0)$ can be computed to additive error $\epsilon$ with confidence $1-\kappa$ in time $\mathsf{poly}(n,k,1/\epsilon) \cdot \log(1/\kappa)$.
\end{lemma}
While the above lemma is essentially identical to Lemma~3.2 in \cite{lovett2015improved}, it is phrased a little differently in that paper. For the sake of completeness, we reprove this lemma in Appendix~\ref{app:Lem1.2}.\\

%
%
Let $A = \mathsf{supp} (f)$ where $A = \{x_1, \ldots, x_k\}$
with $x_1 =0$.  Let $\mathbf{1}_E : \{0,1\}^n \rightarrow \{0,1\}$ be the corresponding function from Lemma~\ref{lem:lovettlem} and $g : \{0,1\}^n \rightarrow \mathbb{R}$ be defined as $g(x)=  f(x) \cdot (T_{\mu} \mathbf{1}_E)(x)$.
From Lemma~\ref{lem:lovettlem}, we get that $g(x)$ decays exponentially in $|x|$ for $|x| \ge \mu^{-2} \cdot \log k$ where $|x|$ denotes the Hamming weight of $x$.
Let $\mathcal{B}_r(0)$ denote the Hamming ball of radius $r$ around the origin.
Then, the above implies that if we set $r = O_\mu(\log (k/\epsilon))$, then $g$ essentially vanishes outside $\mathcal{B}_r(0)$.
 Next, consider a function $\ell:\{0,1\}^n \rightarrow \mathbb{R}$ where we set $T =\Vert \widehat{\ell} \Vert_{L_1}$, $S_0 = \max_{S \in \mathsf{supp}(\widehat{\ell})} |S|$ such that $\ell(0) =1$ and $|\ell(x)| \le \eta$ for $x \in \mathsf{supp}(f) \cap \mathcal{B}_r(0)$. Then, it follows that (as shown in Section~\ref{sec:maintheorem})
$$
\big| \langle \ell , g \rangle  - f(0 ) \cdot (T_{\mu} \mathbf{1}_E)(0)\big|  \le  \eta + T \cdot e^{-\frac{\mu^2 \cdot r}{2}} .
$$
\ignore{
Thus, we have,
\begin{eqnarray*}
\big| \langle \ell , g \rangle  - \ell(0 ) \cdot g(0)\big| &\le& \bigg| \sum_{x : 0< |x| \le r} \ell (x) \cdot g(x)\bigg| +\bigg| \sum_{x : |x| > r} \ell (x) \cdot g(x)\bigg|,  \\
&\le& \sup_{x \in \mathsf{supp}(f) : 0< |x| \le r} |\ell(x)|  + \sup_{x \in \mathsf{supp}(f) :  |x| > r} |g(x)| \cdot  \sup_{x} |\ell(x)|, \\
&\leq&  k \cdot  \epsilon + T \cdot e^{-\frac{\mu^2 \cdot r}{4}}.
\end{eqnarray*}}
If we set $\eta = \epsilon/16$, then  it just remains to bound the second term. Thus, we seek to construct $\ell: \{0,1\}^n \rightarrow \{0,1\}$ which is $1$ at the origin, at most $\eta$ (in absolute value) for $x \in \mathsf{supp}(f) \cap \mathcal{B}_r(0)$ and $T$, $S_0$ and $|\mathsf{supp}(\widehat{\ell})|$ are as small as possible. In particular, in the above error term, we have two competing parameters, namely $T$ and $r$ i.e. as $r$ increases, the value of $T = \Vert \widehat{\ell} \Vert_{L_1}$ corresponding to the optimal $\ell$, also increases. Thus, it is not immediately obvious if there exists $ \ell (\cdot)$ such that the second error term $T \cdot e^{-\frac{\mu^2 \cdot r}{2}}$ can be made vanishingly small.  However, for a careful choice of $\ell$ (as we discuss shortly), the second term can also be made $\epsilon/16$. Thus, $|\langle \ell , g \rangle - f(0 ) \cdot (T_{\mu} \mathbf{1}_E)(0)| \le \epsilon/8$. Thus, if we approximate $\langle \ell, g \rangle$ (as done in Corollary~\ref{corr:att-compute}) and $(T_{\mu} \mathbf{1}_E)(0)$ (as done in Lemma~\ref{lem:lovettlem}), we obtain an $\epsilon$-approximation to $f(0)$.

We now motivate our construction of the function $\ell(\cdot)$. For this, let $C =\mathsf{supp}(f) \cap \mathcal{B}_r(0)$ and let $C^{\downarrow}$ be the downset generated by $C$. Note that $x \in \{0,1\}^n$ can be identified as the characteristic vector of a subset of $[n]$ and hence, for the following discussion, we alternately identify $ \{0,1\}^n$ with $\mathcal{P}([n])$. We will first start with a suboptimal choice of $\ell$ which will motivate our final construction.

Corresponding to every $z \in C^{\downarrow}$, consider the monomial $\mathsf{AND}_z: \{0,1\}^n \rightarrow \{0,1\}$ defined as
$$
\mathsf{AND}_z (x_1, \ldots, x_n) = \prod_{i : z_i=1} x_i.
$$
We will define $\ell$ to be a linear combination of $\mathsf{AND}_z$ for $z \in C^{\downarrow}$ subject to the constraints
$$
\ell(x) = \begin{cases} 0 \ &\textrm{if } x \in C^{\downarrow} \setminus \{0\}, \\ 1 &\textrm{if} \ x=0.   \end{cases}
$$
Since $\ell$ is a linear combination of $\{\mathsf{AND}\}_{z \in C^{\downarrow}}$, let us assume that $\ell = \sum_{z \in C^{\downarrow}} \alpha_z \cdot \mathsf{AND}_z $. In terms of
the $\boldsymbol{\zeta}$ transform for the poset $C^{\downarrow}$, we can express $\ell$ as  $\ell =  \boldsymbol{\zeta}^T (\sum_{z \in C^{\downarrow}}\alpha_z \cdot \boldsymbol{1}_z)$, where $\boldsymbol{1}_z$ is the indicator function of $z$.
Thus,  $\alpha_z = (\boldsymbol{\mu}^T \ell)(z)$  where $\boldsymbol{\mu}$ is the M\"{o}bius transform for $C^{\downarrow}$. Applying Proposition~\ref{prop:mobius} and the value of $\ell$ on $C^{\downarrow}$, we obtain that $\alpha_z = (-1)^{|z|}$. Thus, $\ell (x) = \sum_{z \in C^{\downarrow} }(-1)^{|z|} \mathsf{AND}_z(x)$. It is not difficult to see that for this choice of $\ell$,  $ \Vert \widehat{\ell} \Vert_{L_1} \le |C^{\downarrow}| \le k \cdot 2^r$ and $\max_{S: \widehat{\ell}(S) \not =0} |S| \le r$. This choice of $\ell$ itself yields non-trivial results. In particular, in an earlier version of this paper, the authors proved Theorem~\ref{thm:1} with $\mu \gtrsim 0.555$, thereby giving the first polynomial time algorithm for noisy population recovery for any $\mu<1$.
However, to prove Theorem~\ref{thm:1} for any $\mu >0$, we require a more refined choice of $\ell$. Henceforth, let us refer to the previous choice of $\ell$ as $\ell_{0}$.  In particular, the bottleneck in our argument comes from the fact that we bound
$\Vert \widehat{\ell_0} \Vert_{L_1}$ by  $k \cdot 2^r$. Instead, if we were able to bound $\Vert \widehat{\ell_0} \Vert_{L_1} \le (1+\delta)^r$ for some $\delta <1$, this would immediately imply improve the lower bound required on $\mu$. If  $\delta>0$ could be made arbitrarily small, then we  obtain Theorem~\ref{thm:1} for all $\mu>0$.

Towards a better choice of $\ell$, we notice that while   $\ell_0(x)=0$ for $x \in C^{\downarrow} \setminus \{0\}$, it suffices to have $|\ell(x)| \le \eta =\epsilon/16$ for $x \in C^{\downarrow} \setminus \{0\}$.
Unfortunately, it is not clear how this relaxed requirement on $\ell$ can be exploited by the above analysis. To circumvent this, we consider a new family of functions $\{\mathsf{AND}_{\delta, z}\}_{z \in C^{\downarrow}}$ defined as follows.
$$
\textrm{For } x \in C^{\downarrow}, \ \mathsf{AND}_{\delta, z}(x) = \mathbf{1}_{x \succeq z} \cdot (1-\delta)^{|x|-|z|},
$$
$$
\textrm{ and }\widehat{\mathsf{AND}_{\delta, z}} \textrm{ is supported on }C^{\downarrow}.
$$
It is not difficult to see that the above conditions uniquely define $\mathsf{AND}_{\delta,z}$. For points in $C^{\downarrow}$, one can view $ \mathsf{AND}_{\delta, z}(\cdot)$ as a \emph{noise attenuated} version of the function $\mathsf{AND}_{z} (\cdot)$ (obtained by setting $\delta=0$).  We now set $\ell (x)= \sum_{z \in C^{\downarrow}} \alpha_z \cdot \mathsf{AND}_{\delta,z}(x)$. Obtaining the coefficients $\{\alpha_z\}_{z \in C^{\downarrow}}$ can be viewed as a sort of \emph{noise attenuated M\"{o}bius inversion}.
The flexibility afforded by the parameter $\delta$ allows us to exploit the relaxed constraints on $\ell$ and bound $\alpha_z $ by $\delta^{|z|} \cdot (1/\eta)^{O(\delta^{-1} \cdot \log \delta^{-1})}$.  To prove this, we combine basic properties of the M\"{o}bius transform on $C^{\downarrow}$ with the robust local inverse from Theorem~\ref{thm:moitra-saks}. Intuitively, since the function $\mathsf{AND}_{\delta,z}(\cdot)$ combines properties of $\mathsf{AND}_z$ with noise attenuation, it is not  surprising that the properties of M\"{o}bius transform and the robust local inverse are useful in bounding $\{\alpha_z\}_{z \in C^{\downarrow}}$. Further, we show that
\begin{eqnarray*}
\Vert \widehat{\ell} \Vert_{L_1} &\le&  k^2 \cdot (1+2\delta)^r \cdot (1/\eta)^{O(\delta^{-1} \cdot \log \delta^{-1})}.
\end{eqnarray*}
This proof of this inequality again uses the structure of $C^{\downarrow}$ as well as bounds on $\Vert \widehat{\mathsf{AND}_{\delta, z}} \Vert_{L_1}$ (this proof is given in Section~\ref{sec:4}). The above bound is incomparable to the  bound of $k \cdot 2^r$ we obtained for the first choice of $\ell = \sum_{z \in C^{\downarrow}} (-1)^{|z|} \cdot \mathsf{AND}_z$. In particular, we pay a dependence on $\eta$ to bound $\Vert \widehat{\ell} \Vert_{L_1}$ whereas the bound on $\Vert \widehat{\ell_0} \Vert_{L_1}$ had no dependence on $\eta$. However, the place where we make a significant gain is that base of the exponential factor in $r$ can be made arbitrarily close to $1$ by choosing a suitably small $\delta>0$.
We summarize the properties of $\ell$ in the next theorem.
~\\
\begin{theorem}\label{thm:MS1} Let $C \subseteq \{0,1\}^n$ be as defined above where $|C| \le k$ and $r = \max_{x  \in C} |x|$.
Given any $\delta, \eta >0$, there exists $\ell : \{0,1\}^n \rightarrow \mathbb{R}$ which is a linear combination of $\{\mathsf{AND}_{\delta, z}\}_{z \in C^{\downarrow}}$ such that
\begin{itemize}
\item $\ell(0) = 1$ and $|\ell(x)| \le \eta$ for $x \in C^{\downarrow} \setminus 0$.
\item $\Vert \widehat{\ell} \Vert_1 \le k^2 \cdot (1+2\delta)^{r}  \cdot (2/\eta)^{\delta^{-1} \cdot \log (2\delta^{-1})}$.
\item $\widehat{\ell}$ is supported on $C^{\downarrow}$ and hence $\max_{S : \widehat{\ell}(S) \not =0} |S| =r$.
\end{itemize}
Further, let $\ell(x)  = \sum_{S \in C^{\downarrow}} \ell_S \cdot \chi_S(x)$. Then, for every $S \in C^{\downarrow}$, $\ell_S$ can be computed in time $\mathsf{poly}(|C^{\downarrow}|, n)$.
\end{theorem}
We compare the above theorem with an analogous result in  Lovett and Zhang~\cite{lovett2015improved} who show the existence of $\ell_{\mathsf{LZ}} : \{0,1\}^n \rightarrow \mathbb{R}$  which satisfies
 \begin{itemize}
 \item $\ell_{\mathsf{LZ}} (0) = 1$ and $\ell_{\mathsf{LZ}} (x)=0$ for $x \in C \setminus 0$.
 \item $\Vert \widehat{\ell_{\mathsf{LZ}} } \Vert_{L_1} \le k \cdot k^{\log r}$ and $\max_{S: \widehat{\ell_{\mathsf{LZ}}}(S) \not = 0} |S| \le \log k$.
 \end{itemize}
 (This result is implied by Propersition 3.6 in their paper.)\\
 We now compare $\ell_{\mathsf{LZ}}$ with the function $\ell$ from Theorem~\ref{thm:MS1}
 \begin{itemize}
 \item $\ell_{\mathsf{LZ}} (x)=0$ for $x \in C \setminus 0$ whereas we achieve the incomparable guarantee of $|\ell(x)| \le \eta$ for $x \in C^{\downarrow} \setminus 0$.
  \item $\widehat{\ell_{\mathsf{LZ}}}$ is supported on a subset of $\mathcal{B}_{\log k}(0)$ whereas $\widehat{\ell} $ is supported on a subset of $\mathcal{B}_{r}(0)$. Thus, in terms of compactness of $\widehat{\ell}$, Lovett and Zhang achieve a superior guarantee.
  \item $\Vert \ell_{\mathsf{LZ}} \Vert_{L_1} \le k \cdot k^{\log r}$ whereas for any $\delta, \eta>0$, we achieve $\Vert \widehat{\ell} \Vert_1 \le k^2 \cdot (1+2\delta)^{r}  \cdot (2/\eta)^{\delta^{-1} \cdot \log \delta^{-1}}$.  Our bound has worse asymptotic dependence on $r$ (and a dependence on $\eta$). However, when the value of $\eta$ and $r$ are eventually plugged in (to $\epsilon/16$ and $r= O_{\mu} (\log (k/\epsilon))$ resp.), the bound on $\Vert \widehat{\ell} \Vert_1 $ remains $k^{O(1)}$ (for a fixed $\mu>0$) whereas the bound on $\Vert \ell_{\mathsf{LZ}} \Vert_{L_1}$ becomes $k^{O(\log \log k)}$. This is the crucial place where we gain over Lovett and Zhang~\cite{lovett2015improved}.
 \end{itemize}
 This concludes the proof overview.  We now give the proof of the main theorem.

\section{Proof of Theorem~\ref{thm:1}}\label{sec:maintheorem}
Recall that we are assuming that $\mathsf{supp}(f) = \{x_1, \ldots, x_k\}$ where $x_1=0$. Also, from our discussion in the preceding section, to prove Theorem~\ref{thm:1}, it suffices to show that $f(0)$ can be approximated to  $\epsilon$ with $(k/\epsilon)^{\tilde{O}(1/\mu)}$ samples and in time $\mathsf{poly}((k/\epsilon)^{\tilde{O}(1/\mu)},n)$. Let $E$ be the set defined in Lemma~\ref{lem:lovettlem} and let $g : \{0,1\}^n \rightarrow \mathbb{R}$ be defined as $g( x) = f(x) \cdot (T_{\mu} \cdot \mathbf{1}_E)(x)$.

 Let $r \ge \mu^{-2} \cdot \log k$ (whose precise value will be fixed later). Let $C = \mathsf{supp}(f) \cap \mathcal{B}_r(0)$. Let $\delta, \eta>0$ whose values will be fixed later and let $\ell: \{0,1\}^n \rightarrow \mathbb{R}$ be the function from Theorem~\ref{thm:MS1} corresponding to the parameters $C$, $r$, $\delta$ and $\eta$. As we have mentioned before, $f(0) \cdot (T_{\mu} \cdot \mathbf{1}_E)(0) \approx \langle \ell, g \rangle$. Thus, our algorithm to approximate $f(0)$ will be to approximate $\langle \ell, g \rangle$ (call the approximation $\widetilde{\langle \ell, g \rangle}$) and $(T_{\mu} \cdot \mathbf{1}_E)(0)$ (call the approximation $\Upsilon$) and return $\widetilde{\langle \ell, g \rangle}/ \Upsilon$.

We first bound the difference between $\langle \ell, g \rangle$ and $(T_{\mu} \cdot \mathbf{1}_E)(0) \cdot f(0)$ in terms of $r$, $k$, $\delta$ and $\eta$.
\begin{claim}\label{clm:errorbound}
$$
\big| \langle \ell, g \rangle - f(0) \cdot (T_{\mu} \cdot \mathbf{1}_E)(0) \big|  \le \eta + \Vert \widehat{\ell} \Vert_{L_1}  \cdot e^{-\frac{\mu^2 \cdot r}{2}}.
$$

\end{claim}
\begin{proof}
Using $\ell(0)=1$ and the definition of $g$,
\begin{eqnarray}
\big| \langle \ell, g \rangle - f(0) \cdot (T_{\mu} \cdot \mathbf{1}_E)(0) \big| &=& \big| \langle \ell, g \rangle - \ell(0) \cdot g(0)\big|  \nonumber \\
&\leq& \sum_{x \in C^{} \setminus 0} \big| \ell(x) \cdot g(x) \big| + \sum_{x \not \in C} \big| \ell(x) \cdot g(x) \big|. \label{eq:fin1}
\end{eqnarray}
Next, we bound the first sum.
\begin{eqnarray}
\sum_{x \in C^{} \setminus 0} \big| \ell(x) \cdot g(x) \big|  &\le& \sup_{x \in C^{} \setminus 0} |\ell(x)| \cdot \sum_{x \in C^{} \setminus 0} \big| g(x) \big| \nonumber \\
&\le& \eta \cdot  \sum_{x \in C^{}} |g(x)|   = \eta \cdot  \sum_{x \in C^{}} |f(x) \cdot (T_{\mu} \mathbf{1}_E)(x)| \nonumber \\
&\le& \eta \cdot  \sum_{x \in C^{}} |f(x) | \le \eta.  \label{eq:fin2}
\end{eqnarray}
In the above, the second inequality  follows  from the property of $\ell$ from Theorem~\ref{thm:MS1}, the third inequality uses $\Vert T_{\mu} \mathbf{1}_E  \Vert_{\infty} \le 1$ and the last inequality uses $\Vert f \Vert_1 =1$.  Next, we bound the second sum.
\begin{eqnarray}
\sum_{x  \not \in C} \big| \ell(x) \cdot g(x) \big|  &\le& \sup_{x } |\ell(x)| \cdot \sum_{x  \not \in C^{} } \big| g(x) \big| \nonumber \\
&\le& \Vert \widehat{\ell} \Vert_{L_1} \cdot  \sum_{x \not \in C^{}} |g(x)|  = \Vert \widehat{\ell} \Vert_{L_1}  \cdot \sum_{x \in C^{}} |f(x) \cdot (T_{\mu} \mathbf{1}_E)(x)| \nonumber \\
&\le&  \Vert \widehat{\ell} \Vert_{L_1}  \cdot  \sum_{x \in C^{}} |f(x) | \cdot e^{-\frac{\mu^2 \cdot r}{2}} \le  \Vert \widehat{\ell} \Vert_{L_1}  \cdot e^{-\frac{\mu^2 \cdot r}{2}}\label{eq:fin3}
\end{eqnarray}
The second inequality uses that for all $x$, $|\ell(x)| \le \Vert \widehat{\ell} \Vert_{L_1}$, the third inequality uses Lemma~\ref{lem:lovettlem} whereas the last inequality uses $\Vert f \Vert_1 = 1$. Plugging (\ref{eq:fin2}) and (\ref{eq:fin3}) in (\ref{eq:fin1}), we obtain the claim. \end{proof}
~\\
Note that $\Vert \widehat{\ell} \Vert_1 \le k^2 \cdot (1+2\delta)^r \cdot (2/\eta)^{\delta^{-1} \cdot \log (2 \delta^{-1})}$. If we set,
\begin{itemize}
\item $\eta = \epsilon/4$,
\item $\delta = \mu^2/16$,
\item $r = (100/\mu^4) \cdot \log (1/\mu) \cdot \log (k/\epsilon)$,
\end{itemize}
then using Claim~\ref{clm:errorbound}, we have $\big| \langle \ell, g \rangle - f(0) \cdot (T_{\mu} \cdot \mathbf{1}_E)(0) \big| \le \epsilon/8$.

Let  $\ell(x) = \sum_{S \in C^{\downarrow}} \ell_S \cdot \chi_S(x)$. Using Theorem~\ref{thm:MS1}, we can assume that we have the complete list $\{\ell_S\}_{S \in C^{\downarrow}}$ in time $\mathsf{poly}(n, |C^{\downarrow}|) = \mathsf{poly}(n, k \cdot 2^r) = \mathsf{poly}((k/\epsilon)^{O_{\mu}(1)}, n)$. Applying Corollary~\ref{corr:att-compute}, using $(k/\epsilon)^{O_{\mu}(1)} \cdot \log (1/\kappa)$ and time $n \cdot (k/\epsilon)^{O_{\mu}(1)} \cdot \log (1/\kappa)$, with confidence $1-\kappa$, we can obtain $\widetilde{\langle \ell, g \rangle }$ such that
$
\big| \widetilde{\langle \ell, g \rangle } - {\langle \ell, g \rangle }\big|  \le \frac{\epsilon}{16}.
$ This implies that
\begin{eqnarray}
\big| \widetilde{\langle \ell, g \rangle } - f(0) \cdot (T_{\mu} \cdot \mathbf{1}_E)(0) \big|  &\le& \big| \widetilde{\langle \ell, g \rangle }  - \langle \ell, g \rangle\big| + \big|  \langle \ell, g \rangle-f(0) \cdot (T_{\mu} \cdot \mathbf{1}_E)(0) \big| \nonumber\\
&\le& \epsilon/16 + \epsilon/16 = \epsilon/8. \nonumber
\end{eqnarray}
Using Lemma~\ref{lem:lovettlem}, we can compute an approximation $\Upsilon$ with confidence $1-\kappa$ such that
$|(T_{\mu} \cdot \mathbf{1}_E)(0) - \Upsilon| \le \epsilon/8$ and $\Upsilon \ge 1/2$ in time $\mathsf{poly}(n,k,1/\epsilon) \cdot \log(1/\kappa)$. Thus,
$$
\bigg| \frac{\widetilde{\langle \ell, g \rangle }}{\Upsilon} - f(0) \cdot \frac{(T_{\mu} \cdot \mathbf{1}_E)(0)}{\Upsilon} \bigg|  \le \frac{\epsilon}{8 \cdot \Upsilon} \le \frac{\epsilon}{4},
$$
where the last inequality uses $\Upsilon \ge 1/2$. Further, with probability $1-\kappa$, we have
$$
\bigg| \frac{(T_{\mu} \cdot \mathbf{1}_E)(0)}{\Upsilon}  - 1 \bigg| \le \epsilon/4.
$$
Thus, with probability $1-2\kappa$,
\begin{eqnarray*}
\bigg| \frac{\widetilde{\langle \ell, g \rangle }}{\Upsilon} - f(0) \bigg| &\le& \bigg| \frac{\widetilde{\langle \ell, g \rangle }}{\Upsilon} - f(0) \cdot \frac{(T_{\mu} \cdot \mathbf{1}_E)(0)}{\Upsilon} \bigg| + \bigg| f(0) - f(0) \cdot \frac{(T_{\mu} \cdot \mathbf{1}_E)(0)}{\Upsilon} \bigg| \\
&\leq& \frac{\epsilon}{4} + f(0) \cdot  \bigg| 1 -   \frac{(T_{\mu} \cdot \mathbf{1}_E)(0)}{\Upsilon} \bigg| \le \frac{\epsilon}{4} + f(0) \cdot \frac{\epsilon}{4} \le \frac{\epsilon}{2}.
\end{eqnarray*}
This concludes the proof of Theorem~\ref{thm:1}.

 \section{Proof of Claim~\ref{clm:att-compute}}\label{sec:att-compute}
 We begin by restating Claim~\ref{clm:att-compute}.
  \begin{claim*}
 Let $g : \{0,1\}^n \rightarrow \mathbb{R}$ be defined as $g(x) = f(x) \cdot (T_{\mu} \mathbf{1}_{E}) (x)$. For $S \subset [n]$, $\widehat{g}(S)$ can be computed to additive accuracy $\epsilon$ with probability $1-\kappa$ using $\gamma(\mu, |S|, \epsilon) \cdot \log (1/\kappa)$ samples from $T_{\mu} f$ in time
$n \cdot \gamma(\mu, |S|, \epsilon) \cdot \log (1/\delta)$ where
$\gamma(\mu, |S|, \epsilon) = \mathsf{poly}((1/\mu)^{|S|}, 1/\epsilon)$. Here, we assume that $\mathbf{1}_E(\cdot)$ can be efficiently computed.
 \end{claim*}
 Since $g (x) = f(x) \cdot (T_{\mu} \mathbf{1}_{E}) (x)$, we get
  that
\begin{eqnarray*}
\widehat{g}(S) = \langle (X_S  f) , (T_{\mu} \mathbf{1}_{E}) \rangle = \langle (T_\mu X_S f) , \mathbf{1}_{E} \rangle.
\end{eqnarray*}
 We now make two observations. The first is that for any $S \subseteq [n]$, $T_{\mu, S}$ is a self-adjoint operator. The second is that if $S, S' \subseteq [n]$ are disjoint sets, then
 the operators $X_{S'}$ and $T_{\mu, S}$ commute. Decomposing $T_{\mu} = T_{\mu, S} T_{\mu, \overline{S}}$, we have
 $$
T_{\mu} X_S f = T_{\mu, S} T_{\mu, \overline{S}} X_S f = T_{\mu, S}  X_ST_{\mu, \overline{S}} f  =  T_{\mu, S}  X_ST_{\mu,S}^{-1} T_{\mu} f.
$$
 Thus, we get
 \begin{eqnarray*}
\widehat{g}(S)  = \langle  T_{\mu, S}  X_ST_{\mu,S}^{-1} T_{\mu} f , \mathbf{1}_E \rangle = \mathbf{E}_{z \sim T_{\mu} f} \langle T_{\mu, S}  X_ST_{\mu,S}^{-1} \mathbf{1}_z , \mathbf{1}_E \rangle
\end{eqnarray*}
An easy but crucial fact is the following.
\begin{proposition}
$\langle T_{\mu, S}  X_ST_{\mu,S}^{-1} \mathbf{1}_z , \mathbf{1}_E \rangle$ can be computed in time $\mathsf{poly}(n, 2^{|S|})$.
\end{proposition}
\begin{proof}
To see this, define
$A_{z,S} = \{y: y_{\overline{S}} = z_{\overline{S}}\}$. Observe that
$$
\mathsf{supp} (T_{\mu, S}  X_ST_{\mu,S}^{-1} \mathbf{1}_z) \subseteq A_{z,S} \  \textrm{and} \ |A_{z,S}| = 2^{|S|}.
$$
Further, $T_{\mu, S}  X_ST_{\mu,S}^{-1} \mathbf{1}_z$ can be computed  on any point in $A_{z,S}$ in time $2^{O(|S|)}$.
Using the fact that $\mathbf{1}_E(\cdot)$ can be efficiently evaluated, we conclude that  $\langle T_{\mu, S}  X_ST_{\mu,S}^{-1} \mathbf{1}_z , \mathbf{1}_E \rangle$ can be evaluated in time $\mathsf{poly}(n, 2^{|S|})$.
\end{proof}

Based on the above relation, our procedure to estimate $\widehat{g}(S)$ will be a simple random sampling procedure. Let $M$ be a sufficiently large number (which will be fixed soon).
\begin{itemize}
\item Sample $z_1, \ldots, z_M \sim T_{\mu} f$.
\item Return  $\widetilde{g}_{S}=M^{-1} \cdot \big(\sum_{i=1}^M \langle T_{\mu, S}  X_ST_{\mu,S}^{-1} \mathbf{1}_z , \mathbf{1}_E \rangle \big)$.
\end{itemize}

To establish an upper bound on $M$, we recall the following facts from \cite{lovett2015improved} (Claim~3.5 in \cite{lovett2015improved}).
\begin{claim}
$\Vert T_{\mu,i}\Vert_{1\to 1} =1$ and $\Vert T^{-1}_{\mu,i}\Vert_{1\to 1} =1/\mu.$
\end{claim}
~\\ The above immediately implies
\begin{equation}\label{eq:hyper}
\Vert T_{\mu,S}\Vert_{1\to 1}\leq 1,\quad \quad \Vert T^{-1}_{\mu,S} \Vert_{1\to 1}\leq (1/\mu)^{|S|}.
\end{equation}
Using $\Vert X_S \Vert_{1 \rightarrow 1} \le 1$, this implies that $ \Vert T_{\mu, S}  X_ST_{\mu,S}^{-1} \mathbf{1}_z \Vert_1 \le (1/\mu)^{|S|}$.
%
%
%
%
%
%
$$
\langle T_{\mu, S}  X_ST_{\mu,S}^{-1} \mathbf{1}_z , \mathbf{1}_E \rangle \le \Vert  T_{\mu, S}  X_ST_{\mu,S}^{-1} \mathbf{1}_z \Vert_1 \le  (1/\mu)^{|S|}.
$$
An application of Chernoff bound yields that if $M = \mathsf{poly}(1/\epsilon, 1/|\mu|^{|S|}) \cdot \log (1/\kappa)$, then with probability $1-\kappa$, $|\widetilde{g}(S) - \widehat{g}(S)| \le \epsilon$.
\noindent
\section{Proof of Theorem~\ref{thm:MS1}}\label{sec:4}
Towards the proof of Theorem~\ref{thm:MS1}, we first recall the following basic facts about $\mathsf{AND}_{z}(\cdot)$.
\begin{proposition}\label{prop:specAND}
For any $z \in \{0,1\}$, the function $\mathsf{AND}_{z} : \{0,1\}^n \rightarrow \{0,1\}$ satisfies the following:
\begin{itemize}
\item $\widehat{\mathsf{AND}_{z}}$   is supported on the subsets of $\{i: z_i =1\}$,
\item and $\Vert \mathsf{AND}_{z} \Vert_{L_1} =1$.
\end{itemize}
\end{proposition}
Recall that $C = \{x_1, \ldots, x_k\}  \subseteq \{0,1\}^n$ where $|x_i| \le r$ (we are assuming that the size of the set $C$ is $k$ as opposed to at most $k$).   
The next proposition proves important structural properties of the function $\mathsf{AND}_{\delta, z}(\cdot)$.
\begin{proposition}\label{prop:noisyand}
Let $0 \le \delta \le 1$. Then, for any point $z \in C^{\downarrow}$, there exists $\mathsf{AND}_{\delta, z}: \{0,1\}^n \rightarrow \mathbb{R}$, with the following properties:
\begin{itemize}
\item For $y \in C^{\downarrow}$, $\mathsf{AND}_{\delta, z}(y) = \mathbf{1}_{y \succeq z} \cdot (1-\delta)^{|y|-|z|} $,
\item $\widehat{\mathsf{AND}_{\delta, z}} (S) \not =0$ only if $S \in C^{\downarrow}$.
\item $\Vert \widehat{\mathsf{AND}_{\delta, z}}\Vert_{L_1} \leq k \cdot (1+\delta)^{r-|z|}$.
\end{itemize}
\end{proposition}
 \begin{proof} Let  $\mathsf{Sym}_j : \mathbb{R}^n \rightarrow \mathbb{R}$ as the elementary symmetric polynomial of degree $j$. We first construct the function $\mathsf{AND}_{\delta, 0}$ i.e.~the function $\mathsf{AND}_{\delta, z}$ where $z$ is the origin. Towards constructing $\mathsf{AND}_{\delta,0}$, we  define the function $h_{\delta,0} : \{0,1\}^n \rightarrow \mathbb{R}$ as
 $$
 h_{\delta,0}(y) = \sum_{j=0}^r (-\delta)^j  \cdot \mathsf{Sym}_j (y) = \sum_{j=0}^r \sum_{S \in \binom{[n]}{j}} (-\delta)^j  \cdot \mathsf{AND}_S (y).
 $$
Thus, for any $y \in \{0,1\}^n$, $|y| \le r$,
$$
h_{\delta,0}(y) = \sum_{j=0}^r (-\delta)^j  \binom{|y|}{j} = (1-\delta)^{|y|}.
$$
Next, observe that if $S \not \in C^{\downarrow}$, $\mathsf{AND}_S(y) =0$.  We define $\mathsf{AND}_{\delta, 0} : \{0,1\}^n \rightarrow \mathbb{R}$ as
$$
\mathsf{AND}_{\delta, 0}(y) = \sum_{j=0}^r \sum_{S \in C^{\downarrow} : |S| = j} (-\delta)^j  \cdot \mathsf{AND}_S (y).
$$
In comparison to $h_{\delta,0}(y)$, the only terms dropped in $\mathsf{AND}_{\delta, 0}(y) $ are $\mathsf{AND}_S (y)$ for $S \not \in C^{\downarrow}$. Thus,  for $y \in C^{\downarrow}$,
$$
\mathsf{AND}_{\delta, 0}(y) = \sum_{j=0}^r (-\delta)^j  \binom{|y|}{j} = (1-\delta)^{|y|}.
$$
Thus, this satisfies the first requirement. For the second requirement, we observe that $\widehat{\mathsf{AND}_S }$ is supported on $S^{\downarrow}$. Since $C^{\downarrow}$ is closed under downward closure, we get that $\widehat{\mathsf{AND}_{\delta, 0}}$ is supported on $C^{\downarrow}$.  For the final item, note that
$$
\Vert \widehat{\mathsf{AND}_{\delta, 0}} \Vert_{L_1} \le \sum_{j=0}^r \sum_{S \in C^{\downarrow} : |S| = j} \delta^j  \cdot \Vert \widehat{\mathsf{AND}_S}\Vert_{L_1} \le  \sum_{j=0}^r \sum_{S \in C^{\downarrow} : |S| = j} \delta^j.
$$
The last inequality uses Proposition~\ref{prop:specAND}. Note that $|C^{\downarrow} \cap \{S: |S| = j \}| \le k \cdot \binom{r}{j}$. Thus,
$$
\Vert \widehat{\mathsf{AND}_{\delta, 0}} \Vert_{L_1} \le  \sum_{j=0}^r \sum_{S \in C^{\downarrow} : |S| = j} \delta^j \le  \sum_{j=0}^r  \binom{r}{j} \cdot k \cdot \delta^j = k (1+\delta)^r.
$$
This finishes the construction of $\mathsf{AND}_{\delta,0}$. For $z \in C^{\downarrow} \setminus \{0\}$,
let $\mathcal{I}_z = \{i : z_i =1\}$. Define $\mathsf{AND}_{\delta, 0, \mathcal{I}_z} : \mathbb{R}^{n \setminus \mathcal{I}_z} \rightarrow \mathbb{R}$ as the function $\mathsf{AND}_{\delta,0}$ when the ambient dimensions are restricted to $[n] \setminus \mathcal{I}_z$. Note that correspondingly, we also project $C^{\downarrow}$ to the coordinates $[n] \setminus \mathcal{I}_z$.
$$
\mathsf{AND}_{\delta, z} (y) = \mathsf{AND}_{z} (y) \cdot \mathsf{AND}_{\delta, 0, \mathcal{I}_z}(y).
$$
First, by definition of $\mathsf{AND}_{\delta, 0, \mathcal{I}_z}(y)$, it  follows that for every $y \in C^{\downarrow}$,
$$
\mathsf{AND}_{\delta, 0, \mathcal{I}_z}(y) = (1-\delta)^{|y_{[n] \setminus \mathcal{I}_z}|} = (1-\delta)^{|y|-|z|}.
$$
This implies that $\mathsf{AND}_{\delta, z} (y) = \mathbf{1}_{y \succeq z} \cdot (1-\delta)^{|y|-|z|}$.

Next, by Proposition~\ref{prop:specAND}, $\widehat{\mathsf{AND}_z}$ is supported on the sets $\mathcal{I}_z$ and by the first part of our proof, $\widehat{\mathsf{AND}_{\delta, 0, \mathcal{I}_z}}$ is supported on the projection of $C^{\downarrow}$ to the coordinates in $[n] \setminus \mathcal{I}_z$. This together implies that $\widehat{\mathsf{AND}_{\delta, z}}$ is supported on $C^{\downarrow}$.

Finally, by Proposition~\ref{prop:specAND}, $\Vert \widehat{\mathsf{AND}_z} \Vert_{L_1} = 1$ and by the first part of our proof, $\Vert \widehat{\mathsf{AND}_{\delta, 0, \mathcal{I}_z}} \Vert_{L_1} \le k \cdot (1+\delta)^{r - |z|}$.  Combining these two, we get $\Vert  \widehat{\mathsf{AND}_{\delta, z}} \Vert_{L_1}  \le k \cdot (1+\delta)^{r-|z|}$.
This finishes the proof.
\end{proof}
\begin{proof}[Proof of Theorem~\ref{thm:MS1}]
Recall the matrix $A_{\mu, r} \in \mathbb{R}^{(r+1) \times (r+1)}$ is defined as $A_{\mu,r}(i,j) = \binom{i}{j} \mu^j (1-\mu)^{i-j}$ where the rows and columns are indexed by $0 \le i, j \le r$. Using Theorem~\ref{thm:moitra-saks}, there exists $v \in \mathbb{R}^{r+1}$ such that $\Vert A_{\mu, r} \cdot v - e_0 \Vert \le \eta$ where $\Vert v \Vert_{\infty} \le (2/\eta)^{(1/\delta) \log (2/\delta)}$. Further, $v$ can be computed in time $\mathsf{poly}(r)$.
~\\
We define $\ell (y) = \sum_{z \in C^{\downarrow}} v_{|z|} \cdot \delta^{|z|} \cdot \mathsf{AND}_{\delta, z}(y)$.  First, it easily follows that $\widehat{\ell}$ is supported on $C^{\downarrow}$.
For $y \in C^{\downarrow}$ define the set $\mathsf{Down}_{y,t} = \{ z \in C^{\downarrow} : z \preceq y $ and $|z|=t\}$. Since $C^{\downarrow}$ is closed under downward closure,
$
|\mathsf{Down}_{y,t}| = \binom{|y|}{t}.
$ Note that $\mathsf{AND}_{\delta,z}(y) =0$ for $z \not \in \cup_{0 \le t \le |y|} \mathsf{Down}_{y,t}$.
Thus, for any $y\in C^{\downarrow}$,
$$
\ell (y) = \sum_{z \in C^{\downarrow}} v_{|z|} \cdot \delta^{|z|} \cdot \mathsf{AND}_{\delta, z}(y) =  \sum_{0 \le t \le |y|}  v_{t}  \cdot \delta^t \cdot (1-\delta)^{|y| -t} \cdot \binom{|y|}{t} = (A_{\mu,r} \cdot v)_{|y|}.
$$
Using Theorem~\ref{thm:moitra-saks},  $\ell(0) =1$ and for $x \in C^{\downarrow} \setminus \{0\}$, $|\ell(x)| \le \eta$.  To prove Theorem~\ref{thm:MS1}, all that remains is to bound $\Vert \widehat{\ell} \Vert_{L_1}$.
\begin{eqnarray}
\Vert \widehat{\ell} \Vert_{L_1} &\le& \sum_{z \in C^{\downarrow}} |v_{|z|}| \cdot \delta^{|z|} \cdot \Vert \widehat{\mathsf{AND}_{\delta, z}} \Vert_{L_1} \ \textrm{(follows from definition of $\ell$)}  \nonumber \\
&\le&  \sum_{z \in C^{\downarrow}} |v_{|z|}| \cdot \delta^{|z|} \cdot k \cdot (1+\delta)^{r-|z|} \ \textrm{(using Proposition~\ref{prop:noisyand})} \nonumber \\
&\le& \Vert v \Vert_{\infty} \cdot \sum_{z \in C^{\downarrow}} \delta^{|z|} \cdot k \cdot (1+\delta)^{r-|z|} \nonumber \\
&=& \Vert v \Vert_{\infty} \cdot \sum_{0 \le j \le r} \delta^{j} \cdot k \cdot (1+\delta)^{r-j}  \cdot | \{z \in C^{\downarrow} : |z|=j\} |. \label{eq:est2}
\end{eqnarray}
Since $|C| \le k$ and $C^{\downarrow} \subseteq B_{r}(0)$, it easily follows that
$$
| \{z \in C^{\downarrow} : |z|=j\} | \le k \cdot \binom{r}{j}.
$$
Plugging this in (\ref{eq:est2}), we get
\begin{eqnarray*}
\Vert \widehat{\ell} \Vert_{L_1} &\le& \Vert v \Vert_{\infty} \sum_{0 \le j \le r} \delta^{j} \cdot k \cdot (1+\delta)^{r-j}  \cdot k \cdot \binom{r}{j} \\
&=& k^2 \cdot \Vert v \Vert_{\infty} \cdot (1+2\delta)^r.
\end{eqnarray*}
Using $\Vert v \Vert_{\infty} \le (2/\eta)^{(1/\delta) \cdot \log (2/\delta)}$, we get the final bound on $\Vert \widehat{\ell} \Vert_{L_1}$.
\end{proof}

\section*{Acknowledgments} 
A.D. is grateful to Rocco Servedio for many illuminating conversations about this problem. 

\bibliography{reference}
\appendix
\section{Proof of Lemma~\ref{lem:lovettlem}}\label{app:Lem1.2}
We begin by restating Lemma~\ref{lem:lovettlem}.
\begin{lemma*}
Let $ \{x_1, \ldots, x_k\}$ where $x_1 =0$. Define the set $\mathsf{Far} = \{x_i : d_H(x_1, x_i) \ge (1/\mu^2) \cdot \log k \}$ and define the set $E = \{y \in \{0,1\}^n: d_H(x_1, y) \le d_H(x_i, y) \ \textrm{for all } \ x_i \in \mathsf{Far} \}$.  
\begin{itemize}
\item $(T_{\mu} \mathbf{1}_E)(0) \ge 1/2$.
\item For $x_i \in \mathsf{Far}$, $(T_{\mu} \mathbf{1}_E)(x_i) \le e^{-\frac{1}{2} \cdot \mu^2 \cdot d_H(x_1, x_i)}$.
\end{itemize}
Clearly, the function $\mathbf{1}_E(\cdot)$ can be computed in time $\mathsf{poly}(n,k)$. Further, $(T_{\mu} \mathbf{1}_E)(0)$ can be computed to additive error $\epsilon$ in time $\mathsf{poly}(n,k,1/\epsilon)$.
\end{lemma*}
\begin{proof}
We first lower bound $(T_{\mu} \mathbf{1}_E)(x_1)$. Let $s=\log (k)/\mu^2$. By definition,
\begin{eqnarray*}
(T_{\mu} \mathbf{1}_E)(x_1)=\Pr_{e \sim D_{\mu}} [x_1+e\in E] &=& 1 - \Pr_{e \in D_{\mu}}[x_1 + e \not \in E] \\
&\ge& 1 - \bigg( \sum_{i:d_H (x_1,x_i)\geq s} \Pr_{y \sim x_1 +e} [d_H (x_1,y)\geq d_H (x_i,y)]\bigg)
\end{eqnarray*}
The last inequality follows by the definition of $E$ and union bound.  To lower bound the right hand side, let us define $S_i = \{ j \in [n] : \ x_i $ and $x_i$ differ in the $j^{th}$ coordinate$\}$. If $d_H(x_i, x_1) \ge s$, then $|S_i| \ge s$. For such a point $x_i \in S$,
$$
\Pr_{y \sim x_1 + D_{\mu}} [d_H (x_1,y)\geq d_H(x_i,y)]=\Pr_{e \sim D_{\mu}} \left[\sum_{j\in S_i} e_j \geq |S_i|/2\right]
$$
To bound the above sum, we recall the Chernoff bound.
\begin{proposition*}
Let $X_1, \ldots, X_n$ be $n$ independent $\{0,1\}$ random variables such that $1 \le i \le n$, $\mathbf{E}[X_i] = p$. If $q>p$, then,
$$
\Pr\bigg[X_1 + \ldots + X_n \ge n \cdot q \bigg]  \le \exp \bigg( -\frac{n}{2} \cdot \bigg(\frac qp -1\bigg)^2 \bigg).
$$
\end{proposition*}
Applying the above proposition, we get that
$$
\Pr_{y \sim x_1 + D_{\mu}} [d_H (x_1,y)\geq d_H(x_i,y)] \le \exp \bigg( \frac{-|S_i|}{2} \cdot \mu^2 \bigg) \le \frac{1}{2k}.
$$
This implies that
$$
(T_{\mu} \mathbf{1}_E)(x_1) \ge 1 - \bigg( \sum_{i:d_H (x_1,x_i)\geq s} \Pr_{y \sim x_1 +e} [d_H (x_1,y)\geq d_H (x_i,y)]\bigg) \ge  \frac 12.
$$
We now  upper bound $(T_{\mu} \mathbf{1}_E)(x_i)$ for $x_i \in \mathsf{Far}$. Note that $(T_{\mu} \mathbf{1}_E)(x_i) = \Pr_{e \sim D_\mu} [x_i + e \in E]$.
Note that if $x_i + e \in E$, then $d_H(x_i + e, x_1) \le d_H(x_i +e, x_i)$.  This implies that $\sum_{j \in S_i} e_j \ge |S_i|/2$. Applying the Chernoff bound, we have
$$
(T_{\mu} \mathbf{1}_E)(x_i) = \Pr_{e \sim D_\mu} [x_i + e \in E] \le \Pr_{e \sim D_{\mu}} \bigg[ \sum_{j \in S_i} e_j \ge \frac{|S_i|}{2}\bigg] \le e^{-\frac{1}{2} \cdot \mu^2 \cdot d_H(x_1, x_i)}.
$$
The fact that $\mathbf{1}_E(\cdot)$ can be computed in time $\mathsf{poly}(n,k)$ follows from the definition of $E$. Further, since $\mathbf{1}_E(\cdot)$ is computable in time $\mathsf{poly}(n,k)$ and $D_{\mu}$ is samplable in time $\mathsf{poly}(n)$, we immediately get that
$$
(T_{\mu} \mathbf{1}_E)(x_1) = \Pr_{e \sim D_{\mu}} [x_1 + e \in E],
$$
can be approximated to $\epsilon$ in time $\mathsf{poly}(n,k, 1/\epsilon) \cdot \log(1/\kappa)$ with confidence $1-\kappa$.
\end{proof}
\section{Robust local inverse from \cite{moitra2013polynomial}}\label{app:MS}
Recall that the matrix $A_{\mu,n} \in \mathbb{R}^{(n+1) \times (n+1)}$ is defined to be
$$
A_{\mu,n}(i,j) = \binom{i}{j} \cdot \mu^j \cdot (1-\mu)^{i-j},
$$
where $\binom{i}{j}=0$ if $j >i$. Following Moitra and Saks~\cite{moitra2013polynomial}, we now define an $\epsilon$-local inverse.
\begin{definition}
Let $w \in \mathbb{R}^{n+1}$ such that $\Vert A_{\mu,n} \cdot w - e_0 \Vert_{\infty} \le \epsilon$. Such a vector $w$ is said to be an $\epsilon$-local inverse of $A_{\mu,n}$.
\end{definition}

 Further, $\Vert w \Vert_{\infty}$ is defined to be the sensitivity of such a vector. Definition~2.1 from \cite{moitra2013polynomial} defines $\sigma_n(\mu, \epsilon)$ to be
$$
\sigma_{n} (\mu, \epsilon) = \min_{\Vert A_{\mu,n} \cdot w - e_0 \Vert_{\infty} \le \epsilon} \Vert w\Vert_{\infty}.
$$
The next observation states that the $w$ achieving the optimum in the above definition can be found using linear programming. \begin{observation}\label{obs:lp}
Using linear programming, it is  possible to find $w \in \mathbb{R}^{n+1}$ in time  $\mathsf{poly}(n)$ such that
$
\Vert A_{\mu,n} \cdot w - e_0 \Vert_{\infty} \le \epsilon,
$ such that $\Vert w \Vert_{\infty} = \sigma_{n}(\mu, \epsilon)$.
\end{observation}
~\\ We now restate Theorem~2.2 from \cite{moitra2013polynomial} which gives an upper bound on $\sigma_{n}(\mu, \epsilon)$.
\begin{theorem*}
For all positive integers $n$ and $\mu, \epsilon>0$, $\sigma_n(\mu, \epsilon) = (1/\epsilon)^{f(\mu)}$ where $f(\mu) = (1/\mu) \cdot \log (2/\mu)$.
\end{theorem*}
Now choose $\epsilon_0=\frac{\epsilon}{1+\epsilon}$ in this theorem. Let $\alpha_0$ be the zeroth coordinate of $A_{\mu,n} \cdot w$. Note that $ 1 + \frac{\epsilon}{1+\epsilon} \ge \alpha_0 \ge 1- \frac{\epsilon}{1+\epsilon}$. Define $v=w/\alpha_0$. Then the zeroth coordinate of $A_{\mu,n} \cdot w$ is 1. For the other coordinate $i\not=0$, we have:
$$
|(A_{\mu,n} \cdot v)_i|=|(A_{\mu,n} \cdot w)_i|/\alpha_0 \leq \frac{\epsilon}{1+\epsilon}\cdot \bigg(1-\frac{\epsilon}{1+\epsilon}\bigg)^{-1}\leq \epsilon
$$
Also we have: $\Vert v \Vert_{\infty} = (1/\alpha_0) \cdot \Vert w \Vert_{\infty} \le (1-\frac{\epsilon}{1+\epsilon})^{-1}((1+\epsilon)/\epsilon)^{(2/\mu) \cdot \log (1/\mu)}\le (2/\epsilon)^{(2/\mu) \cdot \log (1/\mu)}$. This proves Theorem~\ref{thm:moitra-saks}.

\end{document}